\documentclass[11pt,a4paper,draft]{article}

\usepackage[top=3cm, bottom=4cm, left=3cm, right=3cm]{geometry}

\usepackage{amsmath,amssymb}
\usepackage{amsthm,amscd}
\usepackage[active]{srcltx}

\newcommand{\emi}{({\em i}\,) }
\newcommand{\emii}{({\em ii}\,) }
\newcommand{\emiii}{({\em iii}\,) }

\newcommand{\Z}{\mathbb{Z}}
\newcommand{\R}{\mathbb{R}}
\newcommand{\C}{\mathbb{C}}

\newcommand{\la}{\lambda}
\newcommand{\eps}{\epsilon}

\newcommand{\Qc}{{\cal Q}}
\newcommand{\Hc}{{\cal H}}
\newcommand{\Mc}{{\cal M}}
\newcommand{\Bc}{{\cal B}}
\newcommand{\Sc}{{\cal S}}

\newcommand{\Uc}{{\cal U}}
\newcommand{\Ic}{{\cal I}}

\newcommand{\gb}{\boldsymbol{g}}

\newcommand{\qb}{\boldsymbol{q}}
\newcommand{\scal}[2]{\langle #1| #2\rangle}
\newcommand{\ot}{\otimes}
\DeclareMathOperator{\id}{{\rm id}}
\DeclareMathOperator{\sgn}{{\rm sgn}}
\DeclareMathOperator{\Ad}{{\rm Ad}}

\newtheorem{thr}{Theorem}
\newtheorem{df}[thr]{Definition}
\newtheorem{lm}[thr]{Lemma}
\newtheorem{prop}[thr]{Proposition}

\numberwithin{equation}{section}
\numberwithin{thr}{section}

\begin{document}
\title{Space of quantum states built over metrics of fixed signature}
\author{Andrzej Oko{\l}\'ow }
\date{20 January 2020}

\maketitle
\begin{center}
{\it  Institute of Theoretical Physics, Warsaw University\\ ul. Pasteura 5, 02-093 Warsaw, Poland\smallskip\\
oko@fuw.edu.pl}
\end{center}
\medskip

\begin{abstract}
We construct a space of quantum states and an algebra of quantum observables, over the set of all metrics of arbitrary but fixed signature, defined on a manifold. The construction is diffeomorphism invariant, and unique up to natural isomorphisms. 
\end{abstract}

\section{Introduction}

In the late 70's of the last century, Jerzy Kijowski proposed a construction method of quantum states for field theories, based on projective techniques, and applied it to a scalar field theory \cite{kpt}. In recent years, search for new models of quantum gravity, led to a considerable development of the method \cite{q-nonl,proj-lt-I,proj-lt-II,proj-lt-III}, which allowed to apply the method to a canonical formulation of the teleparallel equivalent of general relativity (GR) \cite{oko-tegr-I,q-suit,ham-nv,q-stat}, and to GR described in terms of the real Ashtekar variables \cite{barb,proj-lqg-I}. Since the method was successfully applied to these two canonical formulations of GR, it is natural to ask if it can yield a space of quantum states for the ADM formulation of GR \cite{adm}, being perhaps the best known canonical formulation of this theory.      

To work well, the Kijowski's method requires a very special choice of degrees of freedom (d.o.f.) on the phase space of a field theory. So far we were not able to find suitable d.o.f. on the ADM phase space. However, in some cases \cite{non-comp} it is possible to employ a simplified version of the method, which uses only ``position'' d.o.f. on the phase space, discarding momentum ones. This simplified method can be easily applied to the ADM formulation of GR, and yields a space of quantum states.

The ``position'' or configuration variable on the ADM phase space is a Riemannian metric, defined on a three-dimensional manifold. Thus the space of quantum states for the ADM formulation of GR mentioned above, is built over the set of all Riemannian metrics on the manifold. It turns out, however, that the construction of this space makes no essential use of the specific signature of the metric, and the specific dimension of the manifold. Consequently, the simplified method provides a space of quantum states, related to metrics {\em of arbitrary but fixed signature, defined on any manifold}. In this paper we will present this general construction.

To construct the space of quantum states over metrics defined on a manifold, we will use d.o.f. labeled by points of the manifold. The construction will treat all the points equally, without distinguishing any of them. In this sense the resulting space of quantum states will be diffeomorphism invariant.

Given a manifold and a signature of metrics on the manifold, we will obtain actually a family of spaces of quantum states. However, for each two members of the family, there will exist a natural isomorphism, which will map one member to the other. This fact will allow us to state that, given a manifold and a signature, the resulting space of quantum states is unique up to natural  isomorphisms.  

Regarding possible applications of the spaces of quantum states provided by the general construction: perhaps each such space can be used to define a sort of quantum (pseudo-)\\ Riemannian geometry, akin to the quantum Riemannian geometry known from Loop Quantum Gravity (see e.g. \cite{rev}). An interesting question is if the space of quantum states built over Riemannian metrics on a three-dimensional manifold, or that constructed over Lorentzian metrics on a four-dimensional manifold, can be used for quantization of GR.

An essential element of the construction of the quantum states over metrics, will be homogeneous spaces of scalar products. In this paper we will show that on every such a space, there exist invariant metrics and an invariant measure, which is unique up to a multiplicative constant. These homogeneous spaces seem to be interesting by themselves and worthy of further study.     

This paper is organized as follows. Section 2 contains an outline of the Kijowski's method of constructing quantum states. In Section 3 we will describe the homogeneous spaces of scalar products and prove the existence and the uniqueness of invariant measures on these spaces. In Section 4 we will construct the space of quantum space, and in Section 5 we will discuss the construction. In Appendix we will present and prove some technical results (which include explicit formulas for invariant metrics and invariant measures on the homogeneous spaces).  

\section{Preliminaries}

\subsection{Outline of the Kijowski's method \label{outl-kij}}

The Kijowski's method requires to choose some d.o.f. on the phase space of a field theory\footnote{A d.o.f. is a real-valued function on the phase space.}, and organize them into a directed set $(\Lambda,\geq)$---each element $\lambda\in\Lambda$ corresponds to a finite number of d.o.f., and $\lambda'\geq\la$ if $\la'$ represents all d.o.f. related to $\la$. Next, one associates with every $\la\in\Lambda$ a Hilbert space $\Hc_\la$. These two steps of the construction, should be done in such a way that the resulting family $\{\Hc_\la\}_{\la\in\Lambda}$ of Hilbert spaces, is extendable to a richer structure, called in \cite{mod-proj} {\em family of factorized Hilbert spaces}, and introduced (under a different name) in \cite{proj-lt-II}.

Here we will present a slightly simplified definition of family of factorized Hilbert spaces, which, however, will be sufficient for our purpose. The simplification is achieved by requiring the directed set $\Lambda$ to be also partially ordered. We will write $\lambda'>\la$ if $\la'\geq \la$ and $\la'\neq \la$.  

\begin{df}
A family of factorized Hilbert spaces is a quintuplet
\[
\Big(\Lambda,\Hc_\lambda,\tilde{\Hc}_{\lambda'\lambda},\Phi_{\lambda'\lambda},\Phi_{\lambda''\lambda'\lambda}\Big)
\]
such that:
\begin{enumerate}
\item $\Lambda$ is a directed and partially ordered set,
\item for every $\lambda\in\Lambda$, $\Hc_\lambda$ is a Hilbert space,
\item for every $\lambda'>\lambda$, $\tilde{\Hc}_{\lambda'\lambda}$ is a Hilbert space, and
\begin{equation}
\Phi_{\lambda'\lambda}:\Hc_{\lambda'}\to\tilde{\Hc}_{\lambda'\lambda}\ot\Hc_\lambda
\label{Phi}
\end{equation}
is a Hilbert space isomorphism;
\item for every $\lambda''>\lambda'>\lambda$,
\begin{equation}
\Phi_{\lambda''\lambda'\lambda}:\tilde{\Hc}_{\lambda''\lambda}\to\tilde{\Hc}_{\lambda''\lambda'}\ot\tilde{\Hc}_{\lambda'\lambda}
\label{Phi-3}
\end{equation}
is a Hilbert space isomorphism such that the following diagram
\begin{equation}
\begin{CD}
\Hc_{\lambda''} @>\Phi_{\lambda''\lambda}>> \tilde{\Hc}_{\lambda''\lambda}\ot\Hc_\lambda\\
@VV\Phi_{\lambda''\lambda'}V       @VV\Phi_{\lambda'' \lambda'\lambda}\ot\id V \\
\tilde{\Hc}_{\lambda''\lambda'}\ot\Hc_{\lambda'}@>\id\ot\Phi_{\lambda'\lambda}>>  \tilde{\Hc}_{\lambda''\lambda'}\ot\tilde{\Hc}_{\lambda'\lambda}\ot\Hc_\lambda
\end{CD}
\label{diagram}
\end{equation}
is commutative;
\item $\tilde{\Hc}_{\lambda'\lambda} $ and $\Phi_{\lambda'\lambda}$ for other pairs $(\lambda',\lambda)$, and $\Phi_{\lambda''\lambda'\lambda}$ for other triplets $(\lambda'',\lambda',\lambda)$, are not defined.
\end{enumerate}
\label{ffHs}
\end{df}

In some cases it can be quite difficult to construct a family of factorized Hilbert spaces (see e.g. \cite{proj-lqg-I}), but once it is obtained, the remaining steps of the construction of quantum states are straightforward  \cite{proj-lt-II,mod-proj}. One denotes by $\Bc_\la$ the $C^*$-algebra of all bounded operators on $\Hc_\la$, and by $\Sc_\la$ the set of all states (i.e. normed positive linear functionals) on $\Bc_\la$. For every $\la'>\la$, the map \eqref{Phi} and the unit operator $\mathbf{1}_{\la'\la}$ on $\tilde{\Hc}_{\la'\la}$, induce a unital injective $*$-homomorphism:
\begin{equation}
\Bc_\lambda\ni a\mapsto \iota_{\lambda'\lambda}(a):=\Phi^{-1}_{\la'\la}\circ(\mathbf{1}_{\la'\la}\ot a)\circ\Phi_{\la'\la}\in\Bc_{\la'}.
\label{BB}
\end{equation}
For each $\la$, one defines a map $\iota_{\la\la}:\Bc_\la\to\Bc_{\la}$ to be the identity map on the algebra. Then the pull-back defined for every $\la'\geq \la$, 
\begin{equation}
\pi_{\la\la'}\equiv\iota^*_{\la'\la}:\Sc_{\la'}\to \Sc_\la,
\label{pi}
\end{equation}
is a surjection. Furthermore, the commutativity of the diagram \eqref{diagram} guarantees that $\{\Sc_\la,\pi_{\la\la'}\}_{\la\in\Lambda}$ is a projective family. The desired space $\Sc$ of quantum states for the field theory, is defined as the projective limit of the family. 

A byproduct of this construction is the inductive family $\{\Bc_\la,\iota_{\la'\la}\}_{\la\in\Lambda}$ of $C^*$-algebras. Its inductive limit $\Bc$ is again a $C^*$-algebra, which can be regarded as the algebra of quantum observables for the field theory.

There is a one-to-one correspondence between all quantum states in $\Sc$ and all states on the $C^*$-algebra $\Bc$ \cite{mod-proj}.       

\subsection{Application of the Kijowski's method to metrics \label{app}}

Let us now sketch briefly, how we are going to use the method outlined above, to construct a space of quantum states over a set of all metrics of fixed signature, defined on a manifold $\Mc$. To construct the space, we will choose rather natural d.o.f. on the set of the metrics---each of these d.o.f. will map a metric to its value at a point of the manifold. Next, for every point $x\in\Mc$, we will choose a measure $d\mu_x$ on the set $\Gamma_x$ of all scalar products, defined on the tangent space $T_x\Mc$, of signature coincident with that of the metrics ($\Gamma_x$ consists of values of all the metrics at $x$ \footnote{This statement will be made precise by Lemma \ref{k-sur}.}). This will give us a Hilbert space $L^2(\Gamma_x,d\mu_x)$ associated with the point. Finite tensor products of Hilbert spaces of this sort, will form a family $\{\Hc_\la\}_{\la\in\Lambda}$ of Hilbert spaces. This family will be extended in a simple and natural way, to a family of factorized Hilbert spaces, which will yield the desired space of quantum states, in the way described in the previous section.

The construction of the space of quantum states over the metrics will be straightforward, except the choice of the measures $\{d\mu_x\}$ on the spaces $\{\Gamma_x\}$. This choice will be described in the next section. 

\section{Choice of measures $\{d\mu_x\}$ \label{ch-dmux}}

In the present section we will select a measure $d\mu_x$ on every space $\Gamma_x$ of scalar products. Since we see no reason to distinguish {\em a priori} any elements of $\Gamma_x$, we will choose such a measure $d\mu_x$ on this space, which (in a well defined sense) will treat all elements of $\Gamma_x$ equally. Similarly, since we see no reason to distinguish {\em a priori} any points of the manifold $\Mc$, the assignment $x\mapsto d\mu_x$ will be chosen in a diffeomorphism invariant manner.      

\subsection{Set of scalar products of fixed signature}

Each measure $d\mu_x$ is supposed to be defined on a set of all scalar products of fixed signature, defined on a vector space. In this section we will describe some properties of such sets.

\subsubsection{Manifold of scalar products \label{Gamma-mfd}}

Let $V$ be a real $n$-dimensional ($0<n<\infty$) vector space.  A scalar product $\gamma$ on $V$ is a real-valued bilinear symmetric form on $V$, which satisfies the following non-degeneracy condition: if $\gamma(v,v')=0$ for every $v'\in V$, then $v=0$. Each scalar product $\gamma$ on $V$ is characterized by its signature $(p,p')$---if $(v_i)_{i=1,\ldots,n}$ is a basis of $V$, orthonormal with respect to $\gamma$, then $p$ is the number of vectors in the basis such that $\gamma(v_j,v_j)=1$, and $p'$ is the number of vectors in the basis such that $\gamma(v_j,v_j)=-1$. 

Let us fix a pair $(p,p')$ of non-negative integers such that $p+p'=n$, and denote by $\Gamma$ the space of all scalar products on $V$ of signature $(p,p')$. The space $\Gamma$ can be treated as {\em a real-analytic manifold} of dimension $n(n+1)/2$.  

To justify this statement, let us consider the set $\Sigma$ of all bilinear symmetric forms on the vector space $V$. Each basis $(v_i)$ of $V$ defines a global coordinate frame $(\sigma_{ij})_{i\leq j}$ on $\Sigma$:   
\begin{equation}
\Sigma\ni \sigma\mapsto\big(\sigma_{ij}(\sigma)\big)_{i\leq j}:=\big(\sigma(v_i,v_j)\big)_{i\leq j}\in\R^{n(n+1)/2},
\label{sij}
\end{equation}
which is a bijective map. This map can be used to ``pull-back'' the topology from $\R^{n(n+1)/2}$ onto $\Sigma$. Obviously, coordinate frames given by all bases of $V$, form an analytic atlas on $\Sigma$. Let us show now that $\Gamma$ is an open subset of $\Sigma$.   

To this end, let us fix a scalar product $\gamma_0\in \Gamma$ and a basis $(v_i)$ of $V$, orthonormal with respect to $\gamma_0$. Given $\sigma\in \Sigma$, let $m_k(\sigma)$ ($k\in\{1,2,\ldots,n\}$) be the principal $k\times k$ minor of the matrix $(\sigma_{ij}=\sigma(v_i,v_j))$. Obviously, each minor $m_k$ is a polynomial of the coordinates $(\sigma_{ij})_{i\leq j}$ given by $(v_i)$,  and thereby a continuous function on $\Sigma$. Since $m_k(\gamma_0)=\pm 1$, there exists an open neighborhood $U$ of $\gamma_0$ in $\Sigma$, such that each minor $m_k(\sigma)$ is non-zero for every $\sigma\in U$. This means in particular that each $\sigma\in U$ is nondegenerate i.e., $\sigma$ is a scalar product on $V$ of signature $(p_\sigma,p'_\sigma)$. Moreover \cite{sign}, 
\[
p'_\sigma=\frac{1}{2}n-\frac{1}{2}\sum_{k=1}^n\sgn\Big(\frac{m_k(\sigma)}{m_{k-1}(\sigma)}\Big),
\]                 
where $m_0(\sigma):=1$. Since each $m_k$ is either positive or negative on $U$, the r.h.s. of the above expression is a constant function on $U$. Consequently, $(p_\sigma,p'_\sigma)=(p,p')$ for every $\sigma\in U$. Therefore $U$ is a subset of $\Gamma$.        
 
We just showed that for each element of $\Gamma$, there exists an open subset of $\Sigma$, which contains the element and is a subset of $\Gamma$. Thus $\Gamma$ is an open subset of $\Sigma$, and thereby a  real-analytic manifold of dimension $n(n+1)/2$.    

Since now global coordinates $(\sigma_{ij})_{i\leq j}$ on $\Sigma$ given by a basis $(v_i)$ of $V$, when restricted to $\Gamma$, will by denoted by $(\gamma_{ij})_{i\leq j}$.   

In Appendix \ref{noncomp} we will show that $\Gamma$ is noncompact and connected. 

\subsubsection{Homogeneous space of scalar products \label{Gamma-hom}}

Let us consider the group $GL(V)$ of all linear automorphisms of $V$. Given basis $(v_i)$ of $V$, each element $g\in GL(V)$ can be represented by numbers $(g^i{}_j)_{i,j=1,\ldots,n}$ such that
\[
gv_j=g^i{}_jv_i,
\]
where $gv_j$ denotes the action of $g$ on the vector $v_j$. It is clear that the map
\[
GL(V)\ni g\mapsto (g^i{}_j)_{i,j=1,\ldots,n}\in \R^{n^2}
\]
defines a global coordinate frame on $GL(V)$. The atlas consisting of coordinate frames given by all bases of $V$, makes $GL(V)$ a real-analytic manifold and a Lie group \cite{waw}. 

There exists a {\em natural action} of the Lie group $GL(V)$ on the space $\Gamma$ of scalar products:
\begin{equation}
GL(V)\times\Gamma \ni(g,\gamma)\mapsto g\gamma:=g^{-1*}\gamma\in \Gamma,
\label{g-gamma}
\end{equation}
where $g^{-1*}$ is the pull-back: 
\[
(g^{-1*}\gamma)(v,v')=\gamma(g^{-1}v,g^{-1}v'),\quad v,v'\in V. 
\]
For every $g,g'\in GL(V)$ and every $\gamma\in\Gamma$ 
\begin{align*}
g'(g\gamma)&=(g'g)\gamma, & e\gamma&=\gamma
\end{align*}
---$e$ denotes here the identity of the group.

The action \eqref{g-gamma} is also {\em transitive} i.e., for every $\gamma,\gamma'\in\Gamma$ there exists $g\in GL(V)$ such that
\begin{equation}
\gamma'=g\gamma.
\label{g'=gg}
\end{equation}
Indeed, let $(v_i)$ be a basis  of $V$, orthonormal with respect to the scalar product $\gamma$, and let $(v'_i)$ be a basis  of $V$, orthonormal with respect to $\gamma'$. Assume that the ordering of elements of the bases is such that  $\gamma'(v'_i,v'_i)=\gamma(v_i,v_i)$. There exists an element $g\in GL(V)$ such that $g(v'_i)=v_i$. Then 
\[
\gamma'(v'_i,v'_j)=\gamma(v_i,v_j)=\gamma(g^{-1}(v'_i),g^{-1}(v'_j))
\]
and \eqref{g'=gg} follows. 

It is easy to see that the map \eqref{g-gamma}, when expressed in terms of coordinates $(\gamma_{ij})_{i\leq j}$ and $(g^i{}_j)$ (defined by the same basis of $V$), is a rational function of non-zero denominator. This implies that the map is {\em analytic}.

The properties of the space $\Gamma$ and the action \eqref{g-gamma} described above, allow us to conclude that the pair $(GL(V),\Gamma)$ equipped with the map \eqref{g-gamma}, is a {\em homogeneous space} (see e.g. \cite{waw}). It is isomorphic to the homogeneous space $GL(V)/H$, where $H$ is the isotropy group of a scalar product $\gamma_0\in \Gamma$:
\begin{equation}
H:=\{\ g\in GL(V)\ | \ g\gamma_0=\gamma_0\ \}.
\label{H}
\end{equation}

Note finally that $GL(V)$ is isomorphic to $GL(\dim V,\R)$, and $H$ is isomorphic to the (pseudo-)orthogonal group $O(p,p')$, where $(p,p')$ is the signature of scalar products constituting the space $\Gamma$. Thus the homogeneous space $(GL(V),\Gamma)$ of scalar products is isomorphic to $GL(\dim V,\R)/O(p,p')$.

\subsection{Invariant measure on $\Gamma$---existence and uniqueness}


From now on to the end of the paper, unless stated otherwise, a measure will mean a regular Borel measure on a locally compact Hausdorff space (see e.g. \cite{cohn}). If $Y$ is such a space, then the symbol $C_c(Y)$ will denote the linear space of all real-valued continuous functions on $Y$ of compact support. 

We know that $\Gamma$ is an open subset of the space $\Sigma$ of all symmetric bilinear forms on $V$, and that the map \eqref{sij} is a homeomorphism. $\Gamma$ is then a locally compact Hausdorff space as being homeomorphic to an open subset of some $\R^N$.

Let $Y$ and $Y'$ be locally compact Hausdorff spaces and let $\varphi:Y\to Y'$ be a homeomorphism. If $d\mu$ is a regular Borel measure on $Y$ then by virtue of the Riesz representation theorem (see e.g. \cite{cohn}) there exist a unique regular Borel measure $\varphi_\star d\mu$ on $Y'$ such that for every $f\in C_c(Y')$, 
\[
\int_{Y'}f (\varphi_\star d\mu)=\int_Y (\varphi^\star f)\, d\mu,
\]     
where $\varphi^\star f$ is the pull-back of the function $f$: $(\varphi^\star f)(y)=f(\varphi(y))$. Below the measure $\varphi_\star d\mu$ will be called {\em pushforward measure}.

Let a pair $(G,Y)$, where $G$ is a Lie group, and $Y$ a (real-analytic) manifold,  equipped with a map
\begin{equation}
(G,Y)\ni (g,y)\mapsto gy\in Y, 
\label{g-y}
\end{equation}
be a homogeneous space. We will denote by $\bar{g}$ the diffeomorphism
\[
Y\ni y\mapsto gy\in Y
\]
given by $g\in G$---in the case of the homogeneous space $(GL(V),\Gamma)$
\begin{equation}
\bar{g}=g^{-1*}.
\label{bg-g*}
\end{equation}
We say that a measure $d\mu$ on $Y$ is {\em invariant} if for every $g\in G$ 
\begin{equation}
\bar{g}_\star d\mu=d\mu.
\label{inv-df}
\end{equation}

At the very beginning of Section \ref{ch-dmux}, we stated rather loosely that the measure $d\mu_x$ on $\Gamma_x$ should treat equally all points of this space. We know already that $\Gamma_x$ is a homogeneous space. Since the group action on every homogeneous space is transitive, we can make precise our intention regarding $d\mu_x$: it should be an {\em invariant measure} on $\Gamma_x$.  

In this section we will show that \emi on every homogeneous space $(GL(V),\Gamma)$ of scalar products, {\em there exists an invariant measure}, and that \emii the measure is {\em unique up to a positive multiplicative constant}.

To reach the two goals just set, it is necessary to recall some notions and facts regarding Haar measures on Lie groups. Let $d\mu_{H}$ be a left-invariant Haar measure on a Lie group $G$, and let $r_g:G\to G$, $g'\mapsto g'g^{-1}$, be a map defined by the right action of an element $g\in G$. Then there exists a positive real number $\Delta^G(g)$ such that \cite{waw,inv-meas}
\begin{equation}
r_{g\star}d\mu_H=\Delta^G(g)\,d\mu_H.
\label{rg-Haar}
\end{equation}
The map $G\ni g\mapsto \Delta^G(g)\in \R_{+}$ is called right-hand modulus \cite{inv-meas}, or modular function \cite{waw} of the group $G$ (here $\R_+$ is the set of all positive real number). For every Lie group $G$ its modulus $\Delta^G$ satisfies the following equation \cite{waw}: 
\begin{equation}
\Delta^G(g)=\big|\det \Ad(g^{-1}) \big|,
\label{mod-Ad}
\end{equation}
where $\Ad$ denotes the adjoint representation of the group $G$ on its Lie algebra. The group $G$ is said to be unimodular if $\Delta^G=1$. 

\begin{prop}
Suppose that a Lie group $G$ and its closed Lie subgroup $H$ are unimodular. Then there exists a (non-zero) invariant measure on the homogeneous space $G/H$. The measure is unique up to a positive multiplicative constant.  
\label{uni-m}
\end{prop}
\begin{proof}
Theorem 7.4.1 in \cite{inv-meas} ensures that there exists a (non-zero) invariant measure\footnote{Let $Y$ be a locally compact Hausdorff space. In \cite{inv-meas} a measure on $Y$ is defined as a positive linear functional on $C_c(Y)$. By the Riesz representation theorem \cite{cohn} such a functional can be expressed as an integral given by a unique regular Borel measure on $Y$. On the other hand, any regular Borel measure on $Y$ defines a positive linear functional on $C_c(Y)$. Moreover, the invariance of the functional meant in \cite{inv-meas} coincides with the invariance of the corresponding measure defined by \eqref{inv-df}. Thus Theorem 7.4.1 in \cite{inv-meas} can be treated as one concerning regular Borel measures on $G/H$, being a locally compact Hausdorff space \cite{inv-meas}.} on $G/H$ if and only if the modulus $\Delta^G$ restricted to $H$, coincides with the modulus $\Delta^H$ of the group $H$. The theorem implies also, that if $d\mu$ and $d\mu'$ are two invariant measures on $G/H$, then there exists a positive number $c$ such that $d\mu=c\,d\mu'$.

If both $G$ and $H$ are unimodular then $\Delta^G = 1$ and $\Delta^H=1$, and the existence and the uniqueness of an invariant measure on $G/H$, follows by virtue of the theorem.   
\end{proof}

\begin{prop}
There exists a (non-zero) invariant measure on every homogeneous space $\Gamma$ of scalar products. The measure is unique up to a positive multiplicative constant.
\label{E!}
\end{prop}

\begin{proof}
We will treat the cases \emi $\dim V>1$ and \emii $\dim V=1$ separately. The reason is that in the latter case the isotropy group $H$ (see Equation \eqref{H}), is a discrete group isomorphic to $\Z_2$, and an application of Proposition \ref{uni-m} to this case, may not be safe.      

Let $G$ be a Lie group, and $Y$ a manifold. Assume that $(G,Y)$ is a homogeneous space, and that $d\mu$ is an invariant measure on $Y$. If $\phi:Y\to Y'$ is an isomorphism of homogeneous spaces $(G,Y)$ and $(G,Y')$ then $\phi_\star d\mu$ is an invariant measure on $Y'$. Indeed, $\phi$ being an isomorphism of the homogeneous spaces, commutes with the action of the group $G$, $\bar{g}\circ\phi=\phi\circ\bar{g}$,  and consequently 
\begin{equation*}
\bar{g}_\star\phi_\star d\mu=\phi_\star \bar{g}_\star d\mu= \phi_\star d\mu.
\end{equation*}

\paragraph{The case $\dim V>1$} Since each homogeneous space $\Gamma$ is isomorphic to $GL(V)/H$ with $H$ given by \eqref{H}, it is enough to prove the existence and the uniqueness of an invariant measure on every homogeneous space $GL(V)/H$. This latter task will be achieved by showing that both $GL(V)$ and $H$ are unimodular, and applying Proposition \ref{uni-m}.

The group $GL(V)$ is isomorphic to $GL(\dim V,\R)$, while $H$ is isomorphic to either $O(2)$ or $O(1,1)$ or $O(p,p')$ with $p+p'>2$.

The group $GL(\dim V,\R)$ is unimodular \cite{waw}. The group $O(2)$ is compact \cite{hall} and thereby unimodular \cite{waw}.

Regarding the group $O(1,1)$: it can be easily shown that  
\[
SO(1,1)=
\begin{Bmatrix}
\ \pm
\begin{pmatrix}
  \cosh\Psi & \sinh\Psi\\
  \sinh\Psi & \cosh\Psi
\end{pmatrix}\ \Bigg | \ \Psi\in\R \
\end{Bmatrix}.
\]
Therefore $SO(1,1)$ is commutative\footnote{$O(1,1)$ is not commutative---if $h={\rm diag}(1,-1)$, then $h\in O(1,1)$ and $hgh^{-1}=g^{-1}$ for every $g\in SO(1,1)$.}. Since for every element $A$ of the Lie algebra $o(1,1)$, and for every $t\in\R$, the exponential $\exp(tA)$ belongs to $SO(1,1)$, the equality 
\[
g\exp(tA)g^{-1}=\exp(tA)
\]
holds as long as $g\in SO(1,1)$. This means that the operator $\Ad(g)$ is the identity on $o(1,1)$ for every $g\in SO(1,1)$. 

On the other hand, if $g\in O(1,1)$ then $\det g=\pm 1$, and therefore $g^2\in SO(1,1)$. Thus the operator $\Ad(g^2)$ is the identity, and consequently 
\[
1=\det \Ad(g^2)=\big(\det \Ad(g)\big)^2
\]
on the whole $O(1,1)$. This fact together with \eqref{mod-Ad} imply that $O(1,1)$ is unimodular. 

Regarding the group $O(p,p')$ with $p+p'>2$: in this case the Lie algebra $o(p,p')$ is semisimple \cite{hall}, which means that the Killing form $K$ on the algebra is non-degenerate (the Cartan's criterion of semisimplicity). The Killing form is invariant with respect to the adjoint representation $\Ad$ of $O(p,p')$---this fact can be expressed in a basis $(\tau_\alpha)$ of $o(p,p')$ as follows: for every $g\in O(p,p')$ 
\[
K_{\alpha\beta}=K_{\mu\nu}\Ad(g)^\mu{}_\alpha\Ad(g)^\nu{}_\beta.
\]
Calculating the determinant of both sides of this equation, and taking into account the non-degeneracy of the Killing form, we conclude that for every $g\in O(p,p')$
\[
\big(\det\Ad(g)\big)^2=1.
\]
Now to see that $O(p,p')$ is unimodular, it is enough to apply \eqref{mod-Ad}. 

\paragraph{The case $\dim V=1$} Each Lie group is locally compact Hausdorff space, and every left invariant Haar measure on such a group is a regular Borel measure \cite{cohn}.


Let us consider first the signature $(1,0)$. In terms of global coordinates $\gamma_{11}$ on $\Gamma$ and $g^1{}_1$ on $GL(V)$ (given by the same basis of $V$), the action  \eqref{g-gamma} of $GL(V)$ on $\Gamma$ read as follows:
\begin{equation}
(g^1{}_1,\gamma_{11})\mapsto (g^1{}_{1})^{-2}\gamma_{11}.
\label{g-gamma-1}
\end{equation}
Note that this action coincides with the action of the multiplicative group $\R^*_+$ of positive real number on itself (the range of both $(g^1{}_1)^{-2}$ and $\gamma_{11}$ in \eqref{g-gamma-1} is equal $\R^*_+$). Thus the existence and the uniqueness of an invariant measure on $\Gamma$ is guaranteed by the existence and the uniqueness of a (left) invariant Haar measure on $\R^*_+$ (any left invariant Haar is unique up to a positive multiplicative constant (see e.g. \cite{cohn,inv-meas})). 

To prove the existence and the uniqueness of an invariant measure on $\Gamma$ in the case of the signature $(0,1)$, it is enough to note that the homogeneous space of scalar products on $V$  of the signature $(0,1)$, is isomorphic to the space of scalar products on $V$  of the signature $(1,0)$---an example of an isomorphism between the spaces is the map $\gamma\mapsto-\gamma$.    
\end{proof}

In Appendix \ref{meas-ex}, we will construct a (non-Riemannian, in general) metric on every space $\Gamma$ of scalar products, and will show that the metric is invariant with respect to the group action \eqref{g-gamma}. Using the metric we will then derive an explicit expression for an invariant measure on $\Gamma$.

\subsection{Invariant measure on $\Gamma$---properties}

Let $V_0$, $V_1$ and $V_2$  be real vector spaces of the same dimension, and let $\Gamma_i$ ($i=0,1,2$) be the homogeneous space of all scalar products of signature $(p,p')$ on $V_i$ (the signature is the same for all $i$). Any linear isomorphism $l:V_j\to V_i$ defines a pull-back $l^*:\Gamma_i\to \Gamma_j$, being a diffeomorphism between the manifolds $\Gamma_i$ and $\Gamma_j$.

\begin{lm}
Let $d\mu$ be an invariant measure on $\Gamma_0$, and let $l,l':V_1\to V_0$ be linear isomorphisms. Then the pushforward measure $l^*_\star d\mu$ on $\Gamma_1$ coincides with $l^{\prime *}_\star d\mu$.
\label{mu-uniq}  
\end{lm}

\begin{proof}
Note first that $l'\circ l^{-1}\equiv g^{-1}_0$ is an element of $GL(V_0)$ and consequently
\[
l^{\prime*}=l^*\circ \bar{g}_0,
\]
where we used \eqref{bg-g*}. By virtue of this observation and invariance of $d\mu$ 
\begin{equation*}
l^{\prime*}_\star d\mu = l^*_\star \bar{g}_{0\star} d\mu= l^*_\star d\mu.
\end{equation*}
\end{proof}

\begin{lm}
Let $d\mu$ be an invariant measure on $\Gamma_0$ and let $l:V_1\to V_0$ be linear isomorphisms. Then the pushforward measure $l^*_\star d\mu$ on $\Gamma_1$ is invariant. 
\label{lmu-inv}
\end{lm}

\begin{proof}
  Let $g_1$ be an element of $GL(V_1)$. Then $l\circ g^{-1}_1\circ l^{-1}\equiv g^{-1}_0$ is an element of $GL(V_0)$. Using \eqref{bg-g*} we obtain
\[
\bar{g}_1\circ l^*=l^*\circ \bar{g}_0.
\]
Consequently,
\[
\bar{g}_{1\star} (l^*_\star d\mu) = l^*_\star \bar{g}_{0\star} d\mu=l^*_\star d\mu.
\]
\end{proof}

\begin{lm}
Let $l_1:V_1\to V_0$, $l_2:V_2\to V_0$ and $l:V_2\to V_1$ be linear isomorphisms. Suppose that $d\mu$ is an invariant measure on $\Gamma_0$. Then the measure $l^*_{1\star} d\mu$ on $\Gamma_1$ pushed forward by $l^*$ coincides with the measure $l^*_{2\star} d\mu$ on $\Gamma_2$.       
\label{3mu}
\end{lm}

\begin{proof}
Note that $l_1\circ l:V_2\to V_0$ is a linear isomorphism. Consequently,
\[
l^*_\star (l^*_{1\star} d\mu)= (l^*\circ l^*_1)_\star d\mu=(l_1\circ l)^*_\star d\mu=l^*_{2\star} d\mu
\]
---the last equality holds by virtue of Lemma \ref{mu-uniq}.
\end{proof}

The first two lemmas mean that there exists a natural one-to-one relation between invariant measures on two homogeneous spaces $\Gamma_0$, $\Gamma_1$  of scalar products of the same signature, defined on different vector spaces. The last lemma ensures a consistency of such relations in the case of any triplet $\Gamma_0,\Gamma_1,\Gamma_2$ of homogeneous spaces of the same sort:  invariant measures $d\mu_1$ and $d\mu_2$ on, respectively, $\Gamma_1$ and $\Gamma_2$ are related if and only if they are related to the same invariant measure $d\mu_0$ on $\Gamma_0$ \footnote{The properties of the natural relation can be formulated alternatively in the following way. Let $\Ic$ be any set, $\{V_i\}_{i\in\Ic}$ be a family of real vector spaces of the same dimension, and let $\Gamma_i$ be the homogeneous space of all scalar products of signature $(p,p')$ on $V_i$ (the signature is the same for all $i\in \Ic$). Let ${\rm Inv}(\Gamma_i)$ be the set of all invariant measures on $\Gamma_i$. The natural relation on $\bigcup_{i\in \Ic}{\rm Inv}(\Gamma_i)$ can be defined as follows: the measure $d\mu_i\in {\rm Inv}(\Gamma_i)$ is in the natural relation with  $d\mu_j\in {\rm Inv}(\Gamma_j)$ if there exists a linear automorphism $l:V_j\to V_i$ such that $l^*_\star d\mu_i= d\mu_j$. Using invariance of the measures in $\bigcup_{i\in \Ic}{\rm Inv}(\Gamma_i)$ and Lemma \ref{mu-uniq}, \ref{lmu-inv} and \ref{3mu} it is easy to show that \emi the natural relation is an equivalence relation and \emii each equivalence class $[d\mu_j]$ of the relation, contains exactly one element of ${\rm Inv}(\Gamma_i)$ for every $i\in\Ic$. \label{equiv}}.         

\subsection{Diffeomorphism invariant field of invariant measures \label{sec-diff-inv}}

Let $\Mc$ be a manifold. We fix a pair of non-negative integers $(p,p')$ such that $p+p'=\dim \Mc$. As before, we denote by $\Gamma_x$ the space of all scalar products on $T_x\Mc$, of signature $(p,p')$. Obviously, for every $x$ the pair $(GL(T_x\Mc),\Gamma_x)$ is a homogeneous space.

Let us choose a measure $d\mu_x$ on $\Gamma_x$ for every $x\in\Mc$. The resulting assignment $x\mapsto d\mu_x$ can be thought of as a {\em field of measures} or a {\em measure field} on the manifold $\Mc$. In this section we will construct a measure field on $\Mc$, which does not distinguish any point of the manifold.  

To this end, let us choose a point $x_0\in \Mc$ and an {\em invariant} measure $d\mu_{x_0}$ on the homogeneous space $\Gamma_{x_0}$. Then consider the following measure field on $\Mc$:
\begin{equation}
x\mapsto d\mu_{x}:=l^*_{x\star} \,d\mu_{x_0},
\label{diff-inv}
\end{equation}
where $l_x$ is any linear isomorphism from $T_x\Mc$ onto $T_{x_0}\Mc$, and $l^*_x:\Gamma_{x_0}\to \Gamma_x$ the corresponding pull-back. By virtue of Lemma \ref{mu-uniq} the measure $d\mu_x$ does not depend on the choice of the isomorphism $l_x$, and Lemma \ref{lmu-inv} guarantees that $d\mu_x$ is an invariant measure on the homogeneous space $\Gamma_x$. Thus \eqref{diff-inv} is a field of invariant measures. 

Note that the measure field \eqref{diff-inv} is constructed by means of the natural relation\footnote{Using the description of the natural relation introduced in Footnote \ref{equiv}, one can say that the values of the measure field \eqref{diff-inv} form the equivalence class $[d\mu_{x_0}]\subset \bigcup_{x\in\Mc}{\rm Inv}(\Gamma_x)$.} between invariant measures, introduced in the previous section: for every $x\in\Mc$, the measure $d\mu_x$ is in the natural relation with the measure $d\mu_{x_0}$. Moreover, the consistency of the natural relation, guaranteed by Lemma \ref{3mu}, ensures that for every $x,x'\in\Mc$, the measures $d\mu_x$ and $d\mu_{x'}$ are in the natural relation. In this sense, the measure field \eqref{diff-inv} does not distinguish any point of the manifold, including $x_0$. As such, the field should be diffeomorphism invariant. Let us then show that, indeed, it is the case.           

Let $x\mapsto d\mu_x$ be any measure field on $\Mc$, and let $\chi:\Mc\to \Mc$ be a diffeomorphism. If $x_1=\chi(x_2)$, then $\chi$ induces the tangent map $T\chi:T_{x_2}\Mc\to T_{x_1}\Mc$ being a linear isomorphism. The corresponding pull-back $T\chi^*$ is a diffeomorphism from $\Gamma_{x_1}$ onto $\Gamma_{x_2}$. Therefore the pull-back can be used to push forward the measure $d\mu_{x_1}$ to a measure $T\chi^*_\star d\mu_{x_1}$ on $\Gamma_{x_2}$. It this way the diffeomorphism $\chi$ transforms the measure field $x\mapsto d\mu_x$ to another measure field
\begin{equation}
x\mapsto (\chi d\mu)_x:=T\chi^*_\star \,d\mu_{\chi^{-1}(x)}.
\label{dmu-trans}
\end{equation}
We will say that the measure field $x\mapsto d\mu_x$ is {\em diffeomorphism invariant} if
\[
(\chi d\mu)_x=d\mu_x
\]
for every $x\in\Mc$ and for every diffeomorphism $\chi$ of $\Mc$.

%

Consider now the measure field \eqref{diff-inv}, any point $x_1\in \Mc$ and any diffeomorphism $\chi$ of $\Mc$. Let $x_1=\chi(x_2)$, and let $l_{x_1}$ and $ l_{x_2}$ be linear isomorphisms used to define the measures, respectively,  $d\mu_{x_1}$ and $d\mu_{x_2}$ via \eqref{diff-inv}. For $i=0,1,2$, denote $V_i\equiv T_{x_i}\Mc$ and 
\begin{align*}
l_1&\equiv l_{x_1}, & l_2&\equiv l_{x_2}, &l&\equiv T\chi:T_{x_2}\Mc\to T_{x_1}\Mc.
\end{align*}
Now it is enough to  apply Lemma \ref{3mu} and Equation \eqref{dmu-trans} to conclude that 
\[
d\mu_{x_1}=T\chi^*_\star d\mu_{x_2}=T\chi^*_\star d\mu_{\chi^{-1}(x_1)}=(\chi d\mu)_{x_1}.
\]
This means that the measure field \eqref{diff-inv} is diffeomorphism invariant.

Thus we reached the goal, set at the very beginning of Section \ref{ch-dmux}---the field \eqref{diff-inv} of invariant measures does not distinguish any point of the manifold $\Mc$, and each measure $d\mu_x$, being a value of the field, treats equally all elements of $\Gamma_x$.


\section{Construction of quantum states and quantum observables over metrics}

We are now ready to construct the desired spaces of quantum states and the related algebra of quantum observables over metrics of the same signature, defined on a manifold.       

\subsection{Directed set of d.o.f.}

Let us recall that the first step of the construction of quantum states by the Kijowski's method, is the choice of a directed set $(\Lambda,\geq)$ of d.o.f.. Here we will describe such a set, which is suitable for application of the method to metrics.   

Let $\Mc$ be a smooth\footnote{Throughout this paper ``smooth'' means ``of class $C^\infty$''.} manifold. Denote by $\Qc(\Mc)$ the set of all smooth metrics on $\Mc$ of signature $(p,p')$, and assume that the set is non-empty\footnote{The set $\Qc(\Mc)$ is empty for some manifolds and signatures e.g. the only two-dimensional compact manifolds, which admit a global metric of signature $(1,1)$, are the torus and the Klein bottle \cite{steenrod}.}. Obviously, if $q$ is a metric belonging to $\Qc(\Mc)$, then its value $q_x$ at the point $x$, is an element of the homogeneous space $\Gamma_x$ (see the beginning of Section \ref{sec-diff-inv}). As mentioned in Section \ref{app}, maps of the following form:
\begin{equation}
\Qc(\Mc)\ni q\mapsto q_x\in \Gamma_x,
\label{kappa}
\end{equation}
given by all points in $\Mc$, will be treated as d.o.f..

\begin{lm}
If the set $\Qc(\Mc)$ is non-empty, then the map \eqref{kappa} is surjective for every $x\in\Mc$.  
\label{k-sur}
\end{lm}

\noindent A proof of this lemma can be found in Appendix \ref{proof}.

Since  d.o.f. \eqref{kappa} are unambiguously labeled by points of $\Mc$, each finite subset of the manifold represents unequivocally a finite set of d.o.f.. It is then natural to define the set $\Lambda$ to be the set of all finite subsets of the manifold $\Mc$. We say that $\lambda'\in \Lambda$ is greater or equal to $\lambda\in\Lambda$,
\[
\lambda'\geq \lambda,
\]
if $\lambda\subset \lambda'$. Obviously $(\Lambda,\geq)$ is a directed and partially ordered set.

\subsection{Family of factorized Hilbert spaces \label{fam}}

In this section we will construct a family of factorized Hilbert spaces (Definition \eqref{ffHs}) over the set $\Qc(\Mc)$. 

Let $x\mapsto d\mu_x$ be an arbitrary diffeomorphism invariant field of invariant measures, defined via \eqref{diff-inv}. This field makes it possible to define a Hilbert space $\Hc_x$ over each $\Gamma_x$:
\[
\Hc_x:=L^2(\Gamma_x,d\mu_x).
\]
As shown in Appendix \ref{H-sep}, this Hilbert space is {\em separable} for every $x$. Suppose that $\lambda=\{x_1,\ldots, x_N\}\in\Lambda$. Let\footnote{The set $\la$ is unordered, thus to define the tensor product at the r.h.s. of \eqref{H-la}, one has to order elements of $\la$. However, every choice of the ordering is equally well suited for our purposes, and nothing essential depends on the choice. Therefore we will neglect this subtlety.}
\begin{equation}
\Hc_\lambda:=\Hc_{x_1}\ot\ldots\ot \Hc_{x_N}.
\label{H-la}
\end{equation}

Now we will extend the family $\{\Hc_\la\}_{\la\in\Lambda}$ to a family of factorized Hilbert spaces. Recall that we write $\la'>\la$, if $\la'\geq \la$ and $\la'\neq \la$.  Note that if $\la'>\la$, then $\la'\setminus \la$ is again an element of $\Lambda$. Let 
\begin{equation}
\tilde{\Hc}_{\la'\la}:=\Hc_{\la'\setminus\la}.
\label{til-H}
\end{equation}
Then for every $\la'>\la$ and for every $\la''>\la'>\la$,  
\begin{align}
\Hc_{\la'}&=\tilde{\Hc}_{\la'\la}\ot\Hc_\la, & \tilde{\Hc}_{\lambda''\la}&=\tilde{\Hc}_{\lambda''\lambda'}\ot\tilde{\Hc}_{\lambda'\lambda}.
\label{HHH}
\end{align}
These equalities allow us to define the isomorphisms $\Phi_{\la'\la}$ and $\Phi_{\la''\la'\la}$ (see \eqref{Phi} and \eqref{Phi-3}) as the identities on, respectively, $\Hc_{\la'}$  and $\tilde{\Hc}_{\lambda''\la}$. It is a simple exercise to show that the quintuplet $(\Lambda,\Hc_\lambda,$ $\tilde{\Hc}_{\lambda'\lambda},\Phi_{\lambda'\lambda},\Phi_{\lambda''\lambda'\lambda})$ just defined, is a family of factorized Hilbert spaces. 

This family yields the desired space $\Sc$ of quantum states and the $C^*$-algebra $\Bc$ of quantum observables, built over the set $\Qc(\Mc)$, as it is described in Section \ref{outl-kij}.  

\section{Discussion \label{disc}}

\subsection{Uniqueness of the space of quantum states built over metrics}

The construction of the space $\Sc$ of quantum states over a set $\Qc(\Mc)$, seems to be natural, except the arbitrary choice of the diffeomorphism invariant field \eqref{diff-inv} of invariant measures. However, the freedom to choose different fields of this sort, has no significant implications.

To justify this statement, let us consider two fields $x\mapsto d\mu_{x}$ and $x\mapsto d\check{\mu}_{x}$ of invariant measures given by \eqref{diff-inv}. Let us then compare step by step objects appearing at each stage of the construction procedure, leading from each measure field to the resulting spaces of quantum states (the procedure is described in Sections \ref{outl-kij} and \ref{fam}). Objects defined with use of the field  $x\mapsto d\check{\mu}_{x}$, will be distinguished by the check $\check{\phantom{a}}$ from their counterparts given by the field $x\mapsto d\mu_{x}$.  

Let $x_0$ be an arbitrary point of the manifold $\Mc$. By virtue of Proposition \ref{E!}, there exists a positive constant $c$ such that $d\check{\mu}_{x_0}=c\,d\mu_{x_0}$. Pushing forward a measure, which appears in \eqref{diff-inv}, commutes with multiplication of the measure by a constant. Therefore for every $x\in\Mc$, 
\begin{equation}
d\check{\mu}_{x}=c\,d\mu_x.
\label{mux-cmux}
\end{equation}
Consequently, for every $\la\in\Lambda$,   
\[
\Hc_\la\ni\Psi\mapsto U_\la(\Psi):=\frac{1}{(\sqrt{c})^{\# \la}}\Psi \in \check{\Hc}_\la
\]
is a {\em natural} Hilbert space isomorphism (here $\#\la$ denotes the number of elements of $\la$). Clearly, if $\la'>\la$, then 
\begin{equation}
U_{\la'}=U_{\la'\setminus \la}\ot U_{\la}
\label{UUU}
\end{equation}
(see Equations \eqref{til-H} and \eqref{HHH}). 

The map $U_\la$ generates an isomorphism of the (unital) $C^*$-algebras $\Bc_\la$ and $\check{\Bc}_\la$:
\begin{equation}
\Bc_\la\ni a\mapsto u_\la(a):=U_\la\circ a\circ U^{-1}_\la\in \check{\Bc}_\la.
\label{BuB}
\end{equation}
Using \eqref{UUU} and the fact, that here the map $\Phi_{\la'\la}$ (see \eqref{Phi}) and its counterpart $\check{\Phi}_{\la'\la}$ are identities, it is easy to convince ourselves that for every $\la'\geq\la$, 
\begin{equation}
\check{\iota}_{\la'\la}=u_{\la'}\circ\iota_{\la'\la}\circ u^{-1}_\la,
\label{ui-iu}
\end{equation}
where the map $\iota_{\la'\la}:\Bc_{\la}\to\Bc_{\la'}$ is given by \eqref{BB} if $\la'>\la$, and is the identity on $\Bc_\la$ if $\la'=\la$. Equation \eqref{ui-iu} means that the family $\{u_\la\}_{\la\in\Lambda}$ of $C^*$-algebra isomorphisms, maps the inductive family $\{\Bc_\la,\iota_{\la'\la}\}_{\la\in\Lambda}$ to the counterpart family $\{\check{\Bc}_\la,\check{\iota}_{\la'\la}\}_{\la\in\Lambda}$. It is not a difficult exercise to show that the family $\{u_\la\}_{\la\in\Lambda}$ induces a natural isomorphism
\begin{equation}
\Uc:\Bc\to\check{\Bc},
\label{UBB}
\end{equation}
where $\Bc$ and $\check{\Bc}$ are the $C^*$-algebras of quantum observables defined respectively by the inductive families\footnote{$\Bc$ is the inductive limit of the family $\{\Bc_\la,\iota_{\la'\la}\}_{\la\in\Lambda}$. A definition of the inductive limit of an inductive family of $C^*$-algebras, needed to show that \eqref{UBB} is an isomorphism, can be found in e.g. \cite{mod-proj}.}.   

On the other hand, for every $\la$, the map \eqref{BuB} defines the pull-back
\[
u^*_\la:\check{\Sc}_\la\to {\Sc}_\la
\]
being an isomorphism between the two spaces. Recall that for every $\la'\geq \la$, the projection $\pi_{\la\la'}:\Sc_{\la'}\to \Sc_{\la}$ is defined as the pull-back \eqref{pi} of $\iota_{\la'\la}$. By virtue of \eqref{ui-iu} the counterpart projection 
\[
\check{\pi}_{\la\la'}:=\check{\iota}^*_{\la'\la}=u^{-1*}_{\la}\circ\iota^*_{\la'\la}\circ u^*_{\la'}= u^{-1*}_{\la}\circ\pi_{\la\la'}\circ u^*_{\la'}.      
\]
This equation allow us to conclude that the family $\{u^*_\la\}_{\la\in\Lambda}$ of isomorphisms, maps the projective family $\{\check{\Sc}_\la,\check{\pi}_{\la\la'}\}_{\la\in\Lambda}$ onto $\{\Sc_\la,\pi_{\la\la'}\}_{\la\in\Lambda}$, and generates an isomorphism 
\[
\Uc^*:\check{\Sc}\to\Sc
\]
between the resulting spaces of quantum states such that 
\[
\big(\Uc^*(\check{{s}})\big)({a})=\check{{s}}\big(\Uc({a})\big),
\]  
for every $\check{{s}}\in\check{\Sc}$ and ${a}\in\Bc$ (both sides of the equation above are well defined, since every element of a space $\Sc$ of quantum states, obtained via the Kijowski's method, defines unambiguously a state on the corresponding $C^*$-algebra $\Bc$ of quantum observables \cite{mod-proj}).

We just showed that \emi any two diffeomorphism invariant fields $x\mapsto d\mu_{x}$ and $x\mapsto d\check{\mu}_{x}$ of invariant measures, given  by \eqref{diff-inv}, are in the simple relation \eqref{mux-cmux}, and that \emii this relation generates in an unambiguous way the isomorphisms $\Uc$ and $\Uc^*$ between, respectively, the resulting algebras $\Bc$ and $\check{\Bc}$ of quantum observables, and the resulting spaces $\check{\Sc}$ and $\Sc$ of quantum states. This fact allow us to state that both the space $\Sc$ of quantum states and the algebra $\Bc$ of quantum observables, built in the present paper over the set $\Qc(\Mc)$, are {\em unique up to natural isomorphisms}.     

\subsection{Other remarks}
  
Note that the space $\Sc$ and the related algebra $\Bc$ exists for every manifold $\Mc$ and every signature $(p,p')$, even if the corresponding space $\Qc(\Mc)$ of smooth metrics is empty.      

The construction of the space $\Sc$ is also very simple, at least in comparison with similar spaces built in \cite{q-stat,proj-lqg-I} for GR. Thanks to this simplicity, it is easy to describe explicitly some elements of $\Sc$. Indeed, let us choose for every $x\in\Mc$ a normed element $\Psi_x$ of $\Hc_x$. The resulting assignment $x\mapsto \Psi_x$ can be thought of as a {\em field of quantum states} on the manifold $\Mc$. Next define for $\la=\{x_1,\ldots,x_N\}\in\Lambda$
\[
\Psi_\la=\Psi_{x_1}\ot\ldots\ot\Psi_{x_N}\in\Hc_\la.
\]  
Then the map
\[
\Bc_\la\ni a\mapsto s_\la(a):=\scal{\Psi_\la}{\,a\,\Psi_\la}\in\C,
\]
where $\scal{\,\cdot\,}{\,\cdot\,}$ is the inner product on $\Hc_\la$, is a pure state on $\Bc_\la$. For every $\la'\geq \la$,
\[
\pi_{\la\la'}s_{\la'}=s_\la,
\]
where $\pi_{\la\la'}$ is given by \eqref{pi}. Consequently, the net $\{s_\la\}_{\la\in\Lambda}$ is an element of the projective limit $\Sc$ (see e.g. \cite{prof-gr}). Thus $\Sc$ contains elements built solely of pure states, whereas in general there is no guarantee that a space of quantum states constructed by the Kijowski's method, contains any element built solely of normal states (i.e. of states given by density operators) \cite{sl-phd,mod-proj}. 

It is also worth to note that the element $\{s_\la\}_{\la\in\Lambda}$ of $\Sc$ constructed above, is given by  the field $x\mapsto \Psi_x$ of quantum states, that is, by a section of the bundle-like set $\bigcup_{x\in\Mc}\Hc_x$. This observation may be perhaps helpful while building a quantum geometry based on the space $\Sc$.   
 


Let us finally comment on the homogeneous space $\Gamma$ of all scalar products of signature $(p,p')$ on a vector space $V$. Recall that in the proof of transitivity of the group action \eqref{g-gamma} presented in Section \ref{Gamma-hom}, we used two bases $(v_i)$ and $(v'_i)$ of the vector space $V$. Without loss of generality we can assume that both bases define the same orientation of $V$---if it is not the case then it is enough to make a simple change $v_1\mapsto -v_1$ to obtain bases of the desired property. This means that $g$ in \eqref{g'=gg} can be always chosen to be a linear automorphism of $V$ of {\em positive determinant}. Consequently, if $GL^+(V)$ denotes the Lie group of all linear automorphisms of $V$ of positive determinant, then $(GL^+(V),\Gamma)$ with the action \eqref{g-gamma}, is also a homogeneous space and is isomorphic to $GL^+(\dim V,\R)/SO(p,p')$, where $GL^+(\dim V,\R)$ is the Lie group of all  $\dim V\times\dim V$ real matrices of positive determinant.

\paragraph{Acknowledgment} I am very grateful to Jan Derezi\'nski, Pawe{\l} Kasprzak, Jerzy Kijowski, Pawe{\l} Nurow\-ski, Piotr So{\l}tan and Jacek Tafel for discussions and help. This work was partially supported by the Polish National Science Centre grant No. 2018/30/Q/ST2/ 00811.

\appendix

\section{Space of scalar products is noncompact and connected \label{noncomp}}

Let $\Gamma$ be the homogeneous space of all scalar products of signature $(p,p')$ on a real vector space $V$. If $(v_i)$ is a basis of $V$, then $\gamma\mapsto\gamma_{ij}(\gamma)=\gamma(v_i,v_j)$ is a smooth\footnote{This is because functions $(\gamma_{ij})_{i\leq j}$ form a global coordinate frame on the real-analytic manifold $\Gamma$ (see Section \ref{Gamma-mfd}).} function on $\Gamma$. Therefore the following map
\begin{equation}
\Gamma\ni \gamma\mapsto \big|\det\big(\gamma_{ij}(\gamma)\big)\big|\in \R
\label{g-det}
\end{equation}
is  continuous. If $\gamma\in\Gamma$ and $\alpha>0$, then $\alpha\gamma\in\Gamma$ again. This simple fact implies that the image of the map \eqref{g-det}, is the set $\R_+$ of all positive real numbers. Since $\R_+$ is noncompact and the map \eqref{g-det} is continuous, the space $\Gamma$ is noncompact.    

We know from the previous section that for any two scalar products $\gamma,\gamma'\in\Gamma$, there exists $g\in GL^+(V)$ such that $\gamma'=g\gamma$. The Lie group $GL^+(V)$ is isomorphic to the Lie group $GL^+(\dim V,\R)$ (of all  $\dim V\times\dim V$ real matrices of positive determinant), and the latter group is  (path) connected \cite{hall}. Thus there exists a continuous curve
\[
[0,1]\ni t\mapsto \xi(t)\in GL^+(V),
\]     
which starts at the identity $e$ of the group, and ends at $g$. The action \eqref{g-gamma} of the group $GL^+(V)$ on $\Gamma$, is an analytic map. Therefore the curve
\[
[0,1]\ni t\mapsto \xi(t)\gamma\in \Gamma
\]  
is continuous, and starts at $\gamma$ and ends at $\gamma'$. This means that $\Gamma$ is (path) connected.

\section{Construction of invariant measures on spaces of scalar products \label{meas-ex}}

Let $\Gamma$ be the homogeneous space of all scalar products of signature $(p,p')$ on a real vector space $V$. In this section we will define an invariant metric on $\Gamma$, and then will derive from the metric an explicit expression for an invariant measure on $\Gamma$.

\subsection{An invariant metric on $\Gamma$}

Recall that if $(v_i)$ is a basis of $V$, then $\gamma\mapsto\gamma_{ij}(\gamma)=\gamma(v_i,v_j)$ is a smooth function on $\Gamma$. Let $(\gamma^{ij})_{i,j=1,\ldots,\dim V}$ be functions on $\Gamma$ such that
\begin{equation}
\gamma^{ik}\gamma_{kj}=\delta^i{}_j.
\label{gg-del}
\end{equation}
These functions can be used to define the following tensor field on $\Gamma$:
\begin{equation}
Q:=\gamma^{ik}\gamma^{jl}d\gamma_{ij}\ot d\gamma_{kl}.
\label{Q-metric}
\end{equation}
A rather obvious but important observation is that $Q$ does not depend on the choice of the basis $(v_i)$ i.e., if $(\check{v}_i)$ is any other basis of $V$, $\check{\gamma}_{ij}:=\gamma(\check{v}_i,\check{v}_j)$ and  $\check{\gamma}^{ik}\check{\gamma}_{kj}=\delta^i{}_j$, then
\[
Q=\check{\gamma}^{ik}\check{\gamma}^{jl}d\check{\gamma}_{ij}\ot d\check{\gamma}_{kl}.
\]
It follows from \eqref{gg-del} that 
\begin{equation}
(d\gamma^{ik})\gamma_{kj}+\gamma^{ik}d\gamma_{kj}=0
\label{dgg}
\end{equation}
and therefore
\[
Q=\gamma_{ik}\gamma_{jl}\,d\gamma^{ij}\ot d\gamma^{kl}.
\]

Let us prove now that $Q$ is a metric on $\Gamma$. Using obvious symmetricity property $\gamma_{ij}=\gamma_{ji}$ and $\gamma^{ij}=\gamma^{ji}$, it is easy to see that for any vector fields $X,X'$ on $\Gamma$,
\[
Q(X,X')=Q(X',X).       
\]
It remains to show that at every point $\gamma_0\in \Gamma$, the value $Q_{\gamma_0}$ is a nondegenerate form on the tangent space $T_{\gamma_0}\Gamma$. To this end let us assume that the basis $(v_i)$, which defines the functions appearing at the r.h.s. of \eqref{Q-metric}, is an orthonormal basis of $\gamma_0$ such that $\gamma_0(v_i,v_i)=1$ for every $i\leq p$. Then
\[
\gamma^{ij}_0\equiv \gamma^{ij}(\gamma_0)=
\begin{cases}
  1 &\text{if $i=j\leq p$ },\\
  -1 & \text{if $i=j>p$ },\\
  0 & \text{if $i\neq j$}
\end{cases}
\]
and consequently
\begin{equation}
Q_{\gamma_0}=\gamma_0^{ik}\gamma_0^{jl}d\gamma_{ij}\ot d\gamma_{kl}=\sum_{i}(\gamma^{ii}_0)^2d\gamma_{ii}\ot d\gamma_{ii}+2\sum_{i<j}\gamma^{ii}_0\gamma^{jj}_0d\gamma_{ij}\ot d\gamma_{ij}.
\label{Q-diag}
\end{equation}
We see that $Q_{\gamma_0}$ appears above in a diagonal form. Taking into account that the functions $(\gamma_{ij})_{i\leq j}$ form a (global) coordinate frame on $\Gamma$ (see Section \ref{Gamma-hom}), we conclude that $Q_{\gamma_0}$ is nondegenerate. Thus $Q$ is indeed a metric on $\Gamma$.

Since the metric $Q$ is defined by a simple and natural expression \eqref{Q-metric} in terms of natural coordinates on $\Gamma$  we will call it {\em natural metric on} $\Gamma$. However, we are not going to insist that this is the only metric on $\Gamma$, which deserves to be called natural.   

Let us note that the only negative components of $Q_{\gamma_0}$ at the r.h.s. of \eqref{Q-diag}, are those of the form $2\gamma^{ii}_0\gamma^{jj}_0$ for $i<p$ and $j\geq p$. There are $pp'$ of such components. Thus the signature of $Q$ is
\begin{equation}
\Big(\frac{p(p+1)+p'(p'+1)}{2},pp'\Big),
\label{Q-sign}
\end{equation}
and the metric $Q$ is Riemannian if and only if $\Gamma$ consists of either positive or negative definite scalar products.

\begin{lm}
Let $V$ and $V'$ be real vector spaces of the same dimension, and let $\Gamma$ and $ \Gamma'$ be spaces of all scalar products on, respectively, $V$ and $ V'$, of the same signature. Suppose that $l:V\to V'$ is a linear isomorphism and let $l^*:\Gamma'\to\Gamma$ be the corresponding pull-back. If $Q$ is the natural metric on $\Gamma$, then the pull-back $l^{*\star}Q$ of the metric defined by $l^*$ is the natural metric $Q'$ on $\Gamma'$.
\label{lm-Q=Q'}
\end{lm}

\begin{proof}
Let $(\gamma_{ij})$ be functions on $\Gamma$ defined by a basis $(v_i)$ of $V$, and, similarly,  let $(\gamma'_{ij})$ be functions on $\Gamma'$ defined by a basis $(v'_i)$ of $V'$. If $lv_i\in V'$ is the value of the isomorphism $l$ at $v_i\in V$, then
\[
lv_i=l^k{}_iv'_k.
\]

In this setting we have
\begin{equation}
l^{*\star}Q=(l^{*\star}\gamma^{ik})(l^{*\star}\gamma^{jl})(l^{*\star}d\gamma_{ij})\ot(l^{*\star}d\gamma_{kl})=(l^{*\star}\gamma^{ik})(l^{*\star}\gamma^{jl})d(l^{*\star}\gamma_{ij})\ot d(l^{*\star}\gamma_{kl}).
\label{g*-Q}
\end{equation}
Let us fix an arbitrary scalar product $\gamma'_0\in\Gamma'$. Then  
\begin{multline}
(l^{*\star} \gamma_{ij})(\gamma'_0)=(\gamma_{ij})(l^*\gamma'_0)=(l^{*}\gamma'_0)(v_i,v_j)=\gamma'_0(lv_i,lv_j)=\\=\gamma'_0(l^{m}{}_iv'_m,l^{n}{}_jv'_n)=l^{m}{}_i l^{n}{}_j\gamma'_0(v'_m,v'_n)=l^{m}{}_i l^{n}{}_j\gamma'_{mn}(\gamma'_0).
\label{l*-gij}
\end{multline}
This implies that
\begin{align}
l^{*\star} \gamma^{ij}&=l^{-1i}{}_m l^{-1j}{}_n\gamma^{\prime mn}, & d(l^{*\star} \gamma_{ij})=l^{m}{}_i l^{n}{}_jd\gamma'_{mn},
\label{lg,dlg}
\end{align}
since the components $(l^{m}{}_i)$ are constant i.e. they do not depend on $\gamma'_0$. To finish the proof it is enough to set these two results to \eqref{g*-Q}.   
\end{proof}   

Assume now that in Lemma \ref{lm-Q=Q'}, $V=V'$ and $l=g^{-1}\in GL(V)$. Then by virtue of \eqref{bg-g*} and the lemma
\begin{equation}
\bar{g}^\star Q=Q,
\label{g*Q=Q}
\end{equation}
for every $g\in GL(V)$. This result means that the natural metric $Q$ is {\em invariant} with respect to the action \eqref{g-gamma} of the group $GL(V)$ on $\Gamma$. 
   
Suppose now that $\Gamma$ and $\Gamma'$ are the spaces of  all scalar products on $V$ of signature, respectively, $(p,p')$ and $(p',p)$. Then the map
\begin{equation}
\Gamma'\ni \gamma\mapsto \theta(\gamma):=-\gamma\in \Gamma,
\label{theta}
\end{equation}
is an isomorphism of the homogeneous spaces $(GL(V),\Gamma')$ and $(GL(V),\Gamma)$. It is straightforward to check that if $Q$ and $Q'$ are the natural metrics on, respectively, $\Gamma$ and $\Gamma'$, then
\[
\theta^\star Q=Q'.
\]

\subsection{An invariant measure on $\Gamma$ defined by the natural metric}

\subsubsection{Definition}

Functions $(\gamma_{ij})_{i\leq j}$ given by a basis $(v_i)$ of $V$, form a global coordinate frame on $\Gamma$ (see Section \ref{Gamma-hom}). Let us label the coordinates by a single index $I$, i.e. let $(\gamma_{ij})_{i\leq j}=(\gamma^I)_{I=1,\ldots,N}$, where $N=\dim \Gamma$. Denote by $\phi:\Gamma\to\R^N$ the map associated with the coordinates $(\gamma^I)$:
\begin{equation}
\Gamma\ni\gamma\mapsto \phi(\gamma):=\big(\gamma^I(\gamma)\big)\in\R^N.
\label{phi}
\end{equation}
Let $(Q_{IJ})$ be components of the natural metric $Q$ on $\Gamma$ in the coordinate frame $(\gamma^I)$.           

Recall that $\Gamma$ is a locally compact Hausdorff space. The following formula defines a linear functional on $C_c(\Gamma)$: 
\begin{equation}
C_c(\Gamma)\ni f\mapsto \omega(f):=\int_{\phi(\Gamma)}\phi^{-1\star}(f\sqrt{|\det{Q_{IJ}}|})\,d\mu_L\in \R
\label{nat-dmu},
\end{equation}  
where $d\mu_L$ denotes the Lebesgue measure on $\R^N$. The functional is well defined, since $\phi^{-1\star}(f\sqrt{|\det{Q_{IJ}}|})$ is a continuous function of compact support defined on $\phi(\Gamma)$, being an open subset of $\R^N$. Moreover, if a function $f\in C_c(\Gamma)$ is non-negative everywhere, then $\phi^{-1\star}(f\sqrt{|\det{Q_{IJ}}|})$ is non-negative everywhere and, consequently, $\omega(f)$ is non-negative. Thus $\omega$ is a positive functional. Now, the Riesz representation theorem \cite{cohn} guarantees that there exists a unique regular Borel measure $d\mu_{Q}$ on $\Gamma$ such that for every $f\in C_c(\Gamma)$,   
\[
\omega(f)=\int_{\Gamma}f\,d\mu_{Q}.
\]
Using a common convention, the measure $d\mu_{Q}$ can be expressed in terms of the coordinates $(\gamma^I)$ as follows:
\[
d\mu_{Q}=\sqrt{|\det{Q_{IJ}}|}\,d\gamma^1d\gamma^2\ldots\, d\gamma^N.
\]  
Just for convenience we will call $d\mu_{Q}$ {\em natural measure on} $\Gamma$. 

Let us show now that the measure $d\mu_{Q}$ does not depend on the choice of the basis $(v_i)$ of $V$, which defines the coordinates $(\gamma^I)$. To this end we will apply the standard theorem concerning change of variables in a Lebesgue integral \cite{meas-thr}:   
\begin{thr}
  Suppose that $T$ is an open subset of $\R^N$ and a map $\Phi:T\to\R^N$ is bijective and continuously differentiable. Then for any measurable set $W\subset T$ and any Borel function $F\in L^1(\R^N)$, one has the equality
\[
\int_{W} F\,d\mu_L=\int_{\Phi^{-1}(W)} (\Phi^\star F)\, |\det \Phi'|\,d\mu_L,
\]    
where $\Phi'$ is the derivative of $\Phi$.  
\label{ch-var-L}
\end{thr}
Let $(\check{\gamma}^I)$ be coordinates on $\Gamma$ given by a basis $(\check{v}_i)$ of $V$, and let $\check{\phi}:\Gamma\to\R^N$ be the corresponding map \eqref{phi}. The transition function
\[
\Phi=\phi\circ\check{\phi}^{-1}:\check{\phi}(\Gamma)\to\phi(\Gamma), 
\]
is of the following form:
\[
\big(\gamma^I(\gamma_0)\big)=\Phi\big(\check{\gamma}^J(\gamma_0)\big)=\big(\Phi^I{}_J\check{\gamma}^J(\gamma_0)\big),
\]
where $\Phi^I{}_J$ are constants. Consequently, 
\[
\det\Phi'=\det(\Phi^I{}_J)
\]
is a constant as well. On the other hand, 
\[
\check{Q}_{IJ}=\Phi^{M}{}_I\Phi^{N}{}_J{Q}_{MN},
\]
where $(\check{Q}_{IJ})$ are components of $Q$ in the coordinate frame $(\check{\gamma}^I)$. This means that
\[
\sqrt{|\det\check{Q}_{IJ}|}=|\det \Phi^I{}_J|\sqrt{|\det{Q}_{IJ}|}.
\] 
Thus by virtue of Theorem \eqref{ch-var-L}, for every $f\in C_c(\Gamma)$, 
\begin{multline*}
\int_{\Gamma}f\,d\mu_{Q}=\int_{\phi(\Gamma)}\phi^{-1\star}(f\sqrt{|\det{Q_{IJ}}|})\,d\mu_L=\\=\int_{\Phi^{-1}(\phi(\Gamma))}\Phi^\star\big(\phi^{-1\star}(f\sqrt{|\det{Q_{IJ}}|})\big)|\det \Phi^I{}_J|\,d\mu_L=\int_{\check{\phi}(\Gamma)}\check{\phi}^{-1\star}(f\sqrt{|\det{\check{Q}_{IJ}}|})\,d\mu_L.
\end{multline*} 
We conclude that, indeed, the natural measure $d\mu_{Q}$ does not depend on the choice of the basis $(v_i)$, which determines the coordinates $(\gamma^I)$ and the components $(Q_{IJ})$, appearing in \eqref{nat-dmu}.   

\subsubsection{Invariance}

\begin{lm}
Let $V$ and $V'$ be real vector spaces of the same dimension, and let $\Gamma$ and $ \Gamma'$ be spaces of all scalar products on, respectively, $V$ and $ V'$, of the same signature. Suppose that $l:V\to V'$ is a linear isomorphism and let $l^*:\Gamma'\to\Gamma$ be the corresponding pull-back. If $d\mu'_{Q'}$ is the natural measure on $\Gamma'$, then the pushforward measure $l^{*}_{\star}d\mu'_{Q'}$ is the natural measure $d\mu_{Q}$ on $\Gamma$.
\label{lm-dmu=dmu'}
\end{lm}

\begin{proof}
Let $(\gamma^I)\equiv(\gamma_{ij})_{i\leq j}$ be global coordinates on $\Gamma$ defined by a basis $(v_i)$ of $V$, and, similarly,  let $(\gamma^{\prime I})\equiv(\gamma'_{ij})_{i\leq j}$ be global coordinates on $\Gamma'$ defined by a basis $(v'_i)$ of $V'$. Denote by $\phi':\Gamma'\to \R^N$ the map, defined by the coordinates $(\gamma^{\prime I})$ analogously to the map $\phi$ given by \eqref{phi}.

Consider now a map $\Phi:\phi(\Gamma)\to\phi'(\Gamma')$ defined as
\begin{equation}
\Phi:=\phi'\circ l^{-1*}\circ \phi^{-1}.
\label{Phi-2}
\end{equation}
We have
\[
\Phi^{-1}\big(\gamma^{\prime J}(\gamma'_0)\big)=\phi(l^{*}\gamma'_0)=\big(\gamma^{I}(l^{*}\gamma'_0)\big)=\big((l^{*\star}\gamma^{I})(\gamma'_0)\big).
\]
It follows from \eqref{l*-gij} that there exist an invertible matrix $(L^J{}_I)$ of constant components such that
\begin{equation}
l^{*\star}\gamma^{J}=L^J{}_I\gamma^{\prime I}
\label{l*-gI}
\end{equation}
and hence
\[
\Phi^{-1}\big(\gamma^{\prime J}(\gamma'_0)\big)=\big(L^J{}_I\gamma^{\prime I}(\gamma'_0)\big).
\]
Consequently,
\begin{equation}
\det \Phi'=(\det(L^I{}_J))^{-1}
\label{det-Phi}
\end{equation}
is constant as well.

Let $Q$ and $Q'$ be the natural metrics on, respectively, $\Gamma$ and  $\Gamma'$. By virtue of Lemma \ref{lm-Q=Q'}  
\begin{multline*}
Q'=l^{*\star}Q=l^{*\star}(Q_{IJ}\,d\gamma^I\ot d\gamma^J)=(l^{*\star}Q_{IJ})\,d(l^{*\star}\gamma^I)\ot d(l^{*\star}\gamma^J)=\\=(l^{*\star}Q_{IJ})L^I{}_ML^J{}_N\,d\gamma^{\prime M}\ot d\gamma^{\prime N},
\end{multline*}
where in the last step we used \eqref{l*-gI}. Consequently, components $Q'_{MN}$ of $Q'$ in the coordinate frame $(\gamma^{\prime I})$ can be expressed as follows:
\[
Q'_{MN}=(l^{*\star}Q_{IJ})L^I{}_ML^J{}_N.
\]
This implies that
\begin{equation}
\sqrt{|\det Q'_{MN}|}=|\det(L^M{}_N)|\,l^{*\star} \sqrt{|\det Q_{IJ}|}.
\label{det-det}
\end{equation}

Now to finish the proof it is enough to express the pushforward measure $l^*_\star d\mu'_{Q'}$ in terms of the components $(Q_{MN})$ and the map $\phi$: for every $f\in C_c(\Gamma)$, 
\begin{multline*}
\int_{\Gamma}f \,(l^*_\star d\mu'_{Q'})=\int_{\Gamma'}(l^{*\star}f)\,d\mu'_{Q'}=\int_{\phi'(\Gamma')}\phi^{\prime -1 \star}\big((l^{*\star}f)\sqrt{|\det Q'_{MN}|}\big)\,d\mu_L=\\=\int_{\phi'(\Gamma')}\phi^{\prime -1 \star}\big(l^{*\star}(f\sqrt{|\det Q_{IJ}|})\big)|\det(L^M{}_N)|\,d\mu_L,
\end{multline*}
where in the last step we used \eqref{det-det}. Let us now change variables in the last integral by means of the map \eqref{Phi-2}---using Theorem \ref{ch-var-L} and Equation \eqref{det-Phi} we obtain
\[
\int_{\Gamma}f \,(l^*_\star d\mu'_{ Q'})=\int_{\phi(\Gamma)}\phi^{-1\star}(f\sqrt{|\det{Q_{IJ}}|})\,d\mu_L=\int_{\Gamma}f\,d\mu_{ Q}.
\]
Thus, indeed, $l^*_\star d\mu'_{ Q'}=d\mu_{ Q}$. 
\end{proof}

Assume now that in Lemma \ref{lm-dmu=dmu'}, $V=V'$ and $l=g^{-1}\in GL(V)$. Then by virtue of Equation \eqref{bg-g*} and the lemma
\begin{equation}
\bar{g}_\star d\mu_{ Q}=d\mu_{ Q},
\label{g*dmu=dmu}
\end{equation}
for every $g\in GL(V)$. This result means that the natural measure $d\mu_{ Q}$ is an {\em invariant} measure on the homogeneous space $\Gamma$.

\subsubsection{Remarks}

Suppose now that $\Gamma$ and $\Gamma'$ are the spaces of  all scalar products on $V$ of signature, respectively, $(p,p')$ and $(p',p)$. It is not difficult to check that if $d\mu_{ Q}$ and $d\mu'_{ Q'}$ are the natural measures on, respectively, $\Gamma$ and $\Gamma'$, then
\[
\theta_\star d\mu'_{ Q'}=d\mu_{ Q},
\]
where $\theta$ is the isomorphism \eqref{theta} of the homogeneous spaces $\Gamma$ and $\Gamma'$.

Let us note that given dimension of the vector space $V$, the components of the natural metric $ Q$ in a coordinate frame $(\gamma_{ij})_{i\leq j}$, expressed in terms of the corresponding functions $(\gamma^{ij})_{i\leq j}$, are of the same form for every signature. Consequently, the natural measure $d\mu_{ Q}$, when expressed by means of $(\gamma_{ij})_{i\leq j}$ and  $(\gamma^{ij})_{i\leq j}$, is of the same form for every signature $(p,p')$ such that $p+p'=\dim V$. For example, for signatures $(1,0)$ and $(0,1)$, the natural measure reads
\[
d\mu_{Q}=|\gamma^{11}|\,d \gamma_{11}=\frac{1}{|\gamma_{11}|}\,d\gamma_{11},
\]
and for signatures $(2,0)$, $(1,1)$ and $(0,2)$,
\[
d\mu_{ Q}=\sqrt{2\big|(\gamma^{11})^3(\gamma^{22})^3+3\gamma^{11}\gamma^{22}(\gamma^{12})^4-3(\gamma^{11})^2(\gamma^{22})^2(\gamma^{12})^2-(\gamma^{12})^6\big|}\,d\gamma_{11}d\gamma_{22}d\gamma_{12}.
\]        

Our last remark concerns the application of the natural measure $d\mu_{ Q}$, to the construction of quantum states over a set $\Qc(\Mc)$ of metrics, defined on a manifold $\Mc$. Suppose that the measure $d\mu_{x_0}$, which generates the measure field \eqref{diff-inv}, is chosen to be the natural measure on $\Gamma_{x_0}$. It follows from Lemma \ref{lm-dmu=dmu'} that then for every $x\in\Mc$, the measure $d\mu_x$ given by \eqref{diff-inv}, is the natural measure on $\Gamma_x$.

\subsection{Other invariant metrics on $\Gamma$ }

Proposition \ref{E!} ensures that an invariant measure on $\Gamma$ is unique up to a positive multiplicative constant. One may wonder whether the natural metric $d\mu_{Q}$ on $\Gamma$, being an invariant metric, is also unique up to a multiplicative constant. Here we will briefly describe a counterexample to this conjecture.   

Let us consider the following one-form on $\Gamma$: 
\[
\alpha:=\gamma^{ij}d\gamma_{ij}=-\gamma_{ij}d\gamma^{ij},
\]
where the last equality holds by virtue of Equation \eqref{dgg}. It is clear, that if the functions $(\check{\gamma}_{ij})$ and $(\check{\gamma}^{ij})$ are defined by a basis $(\check{v}_i)$ of $V$, then
\[
\alpha:=\check{\gamma}^{ij}d\check{\gamma}_{ij}=-\check{\gamma}_{ij}d\check{\gamma}^{ij}.
\]    
We will call $\alpha$ {\em natural one-form on} $\Gamma$. Using \eqref{lg,dlg} it is easy to show that the natural one-form is invariant with respect to the action of $GL(V)$ on $\Gamma$: 
\[
\bar{g}^{\star}\alpha=\alpha,
\]   
for every $g\in GL(V)$. 

For any real number $a$ we define the following tensor field on $\Gamma$: 
\[
Q^a:=Q+a\,\alpha\ot\alpha,
\]
where $Q$ is the natural metric on $\Gamma$. The invariance of the natural metric $Q$ and of the natural one-form $\alpha$ imply that $Q^a$ is a tensor field invariant with respect to the action of the group $GL(V)$ on $\Gamma$: for every $g\in GL(V)$,
\[
\bar{g}^{\star}Q^a=Q^a.
\] 

Let $Q^{-1}$ denote the metric ``inverse'' to the natural metric $Q$, i.e., $Q^{-1}$ is the tensor field on $\Gamma$ of type $\binom{2}{0}$ such that its components $(Q^{IJ})$ in any coordinate frame satisfy $Q^{IK}Q_{KJ}=\delta^I{}_J$, where $(Q_{IJ})$ are components of $Q$ in the same frame. The invariance of $Q$ implies that $Q^{-1}$ is also invariant with respect to the action of the group $GL(V)$. Since $\alpha$ is also invariant, the function $Q^{-1}(\alpha,\alpha)$ on $\Gamma$ is constant. It turns out that its value equals $\dim V$. 

Let
\[
a_0\equiv-\frac{1}{Q^{-1}(\alpha,\alpha)}<0.
\]
It is possible to show (after some calculations) that if $a\neq a_0$, then $Q^a$ is a metric on $\Gamma$. If $a=a_0$, then $Q^a$ is a ``degenerate metric'' (the vector field $\vec{\alpha}:=Q^{-1}(\alpha,\cdot)$ on $\Gamma$ is non-zero everywhere, but if $a=a_0$, then $Q^a(\vec{\alpha},X)=0$ everywhere for every vector field $X$ on $\Gamma$). If $a>a_0$, then the signature of $Q^a$ coincides with that of $Q$ (see \eqref{Q-sign}), if $a<a_0$, then the signature of $Q^a$ reads
\[
\Big(\frac{p(p+1)+p'(p'+1)}{2}-1,pp'+1\Big).
\]        

\section{Separability of $L^2$ spaces defined by invariant measures \label{H-sep}}

The Hilbert space $L^2(\Gamma,d\mu)$, defined by any invariant (regular Borel) measure $d\mu$ on a homogeneous space $\Gamma$ of scalar products, is {\em separable}. To justify this statement, recall that $\Gamma$ is homeomorphic to an open subset of some $\R^N$. Therefore, $\Gamma$ is {\em second-countable} (i.e. there exists a countable basis for its topology). Now, separability of $L^2(\Gamma,d\mu)$ is a consequence of the following two propositions (see, respectively, \cite{cohn} and \cite{spectr-qm}):  

\begin{prop}
Let $Y$ be a locally compact Hausdorff space that has a countable basis for its topology. Then every regular measure on $Y$  is $\sigma$-finite (i.e. $Y$ is a union of a countable number of subsets of finite measure).
\end{prop}

\begin{prop}
Every $\sigma$-finite Borel measure on a second-countable topological space produces a separable $L^2$ space.
\end{prop}

\noindent Thus for any manifold $\Mc$, all Hilbert spaces $\{\Hc_x\}_{x\in\Mc}$ given by any diffeomorphism invariant field  \eqref{diff-inv} of invariant measures, are separable. 

\section{Proof of Lemma \ref{k-sur} \label{proof}}

Assume that a set $\Qc(\Mc)$ of all metrics of signature $(p,p')$ on a manifold $\Mc$, is non-empty. Let us fix an element $\qb\in \Qc(\Mc)$, a point $\bar{x}\in\Mc$ and a scalar product $\gamma\in \Gamma_{\bar{x}}$. We will show that $\qb$ can be transformed to a smooth metric $q$ on $\Mc$ of the same signature such that its value $q_{\bar{x}}$ at $\bar{x}$ coincides with $\gamma$.

It follows from the comment presented in the last paragraph of Section \ref{disc} that there exists $\gb\in GL(T_{\bar{x}}\Mc)$ of positive determinant such that
\begin{equation}
\gb^*\qb_{\bar{x}}=\gamma.
\label{gq-gamma}
\end{equation}


Let $(y^i)$ be a coordinate frame defined on a neighborhood $\Uc$ of $\bar{x}$ such that $y^i(\bar{x})=0$. We assume moreover that the set of values of the coordinates, contains the closed unit ball in $\R^{\dim\Mc}$ centered at $0$.

The equation \eqref{gq-gamma} expressed in the coordinates $(y^i)$ reads
\[
\qb_{\bar{x} ij}\gb^i{}_m\gb^j{}_n=\gamma_{mn}.
\]   
Note now that the matrix $(\gb^i{}_j)$ is an element of the group $GL^+(\dim\Mc,\R)$ of all  $\dim\Mc\times\dim\Mc$ real matrices of positive determinant. This group is (path) connected \cite{hall}---there exists a continuous curve $\xi:[0,1]\mapsto GL^+(\dim\Mc,\R)$, which starts at $(\gb^i{}_j)$ and ends at the identity $e$ of the group.  


A piecewise smooth curve $\zeta$ in the group $GL^+(\dim\Mc,\R)$, which starts at $(\gb^i{}_j)$ and ends at $e$, can be obtained from $\xi$---it is possible to choose a finite number of points on the image of $\xi$, including $(\gb^i{}_j)$ and $e$, and connect each pair of consecutive points by a smooth curve. $\zeta$ is then composed of a finite number of smooth curves $\{\zeta_1,\zeta_2,\ldots,\zeta_N\}$ such that: \emi $\zeta_1$ starts at $(\gb^i{}_j)$, \emii $\zeta_I$, $1<I\leq N$, starts where $\zeta_{I-1}$ ends, and \emiii $\zeta_N$ ends at $e$. 

There exists a smooth non-decreasing function $f:[0,1]\mapsto[0,1]$ such that it is equal $0$ on the interval $[0,\eps]$, and is equal $1$ on the interval $[1-\eps,1]$ for some $\eps\in[0,1/2[$. Using this function one can reparameterize each $\zeta_I$ to obtain so called ``lazy'' curve \cite{baez-path}, that is, a smooth map
\[
[0,1]\ni t\mapsto \tilde{\zeta}_I(t)\in GL^+(\dim\Mc,\R),
\]
such that: \emi it is constant on $[0,\eps]$ and is constant on $[1-\eps,1]$, \emii $\tilde{\zeta}_I$ starts (ends) where $\zeta_I$ does and \emiii the images of both curves coincide.  
  
The composition of ``lazy'' curves is again a ``lazy'' curve \cite{baez-path}. Thus the composition of the ``lazy'' curves $\{\tilde{\zeta}_1,\tilde{\zeta}_2,\ldots,\tilde{\zeta}_N\}$ is ``lazy''. Thereby the composition is a smooth curve
\[
[0,1]\ni t\mapsto \tilde{\zeta}(t)\in GL^+(\dim\Mc,\R),
\] 
which starts at $(\gb^i{}_j)$ and ends at $e$. 

Let $\big(\tilde{\zeta}^i{}_m(t)\big)$ be the components of the matrix $\tilde{\zeta}(t)$ and let $r^2:\Uc\to \R$ be a smooth map defined as follows:
\[
x\mapsto r^2(x):=(y^1(x))^2+(y^2(x))^2\ldots+(y^{\dim \Mc}(x))^2.
\]  
It is easy to realize that the assignment
\[
\Mc\ni x\mapsto q_x=
\begin{cases}
\qb_{xij}\tilde{\zeta}^i{}_m(r^2(x))\tilde{\zeta}^j{}_n(r^2(x))\,dy^n\,dy^m&\text{for $x\in \Uc$ such that $r^2(x)\leq 1$,}\\
\qb_{x}& \text{otherwise,}
\end{cases}
\]
defines a smooth metric $q\in \Qc(\Mc)$ such that $q_{\bar{x}}=\gamma$.



\end{document}